\documentclass[a4paper,12pt]{article}
\usepackage{color}
\usepackage{epsfig}
\usepackage{epstopdf}
\usepackage{subfigure}
\usepackage{amssymb}
\graphicspath{{eps/}}
\usepackage{amsbsy}
\usepackage{amsmath}
\usepackage{amsthm}
\usepackage{color}
\usepackage{booktabs}
\usepackage{graphicx,adjustbox}
\usepackage{algorithm}
\usepackage{algorithmic}
\usepackage{hyperref}
\usepackage{url}
\usepackage[round]{natbib}
\usepackage{pgfplots}
\pgfplotsset{compat=1.14}
\usepackage{mathrsfs}
\usepackage[margin=1in]{geometry}
\usepackage{setspace}
\usepackage[title]{appendix}

\long\def\symbolfootnote[#1]#2{\begingroup%
	\def\thefootnote{\fnsymbol{footnote}}\footnote[#1]{#2}\endgroup}
\newcommand*{\blue}{\textcolor{black}}

\usepackage{epsfig}

\allowdisplaybreaks

\begin{document}
%\begin{frontmatter}

\begin{center}
	
	\begin{spacing}{2}
		{\LARGE Scheduling multiple agile Earth observation satellites \\ with multiple observations}
	\end{spacing}
	
	\large{Xinwei Wang$^{1}$\symbolfootnote[1]{Corresponding author. E-mail: xinwei.wang@qmul.ac.uk (X. Wang), hanchao@buaa.edu.cn (C. Han), roel.leus@kuleuven.be (R. Leus).}, Chao Han$^{2}$ and Roel Leus$^3$
	}
	
	\bigskip
	
	\begin{spacing}{1}
		\normalsize
$^{1}$\blue{School of Engineering and Materials Science, Queen Mary University of London, 327 Mile End Road, London E1 4NS, UK} \\ 
$^{2}$\blue{School of Astronautics, Beihang University, 37 Xueyuan Road, Beijing 100191, China} \\ 
$^{3}$\blue{ORSTAT, KU Leuven, Naamsestraat 69, 3000 Leuven, Belgium}
	\end{spacing}
\end{center}

%Here is the abstract:
\footnotesize \noindent \textbf{Abstract:}
Earth observation satellites (EOSs) are specially designed to collect images according to user requirements. Agile EOSs (AEOSs), with stronger attitude maneuverability, greatly improve the observation capability, while increasing the complexity of scheduling the observations. In this paper, we address the problem of scheduling multiple AEOSs with multiple observations where the objective function aims to maximize the entire observation profit over a fixed horizon. The profit attained by multiple observations for each target is nonlinear in the number of observations. Our model is a specific interval scheduling problem, with each satellite orbit represented as a machine. A column-generation-based framework is developed for this problem, in which the pricing problems are solved using a label-setting algorithm. Extensive computational experiments are conducted on the basis of one of China's AEOS constellations.  The results indicate that our optimality gap is less than $3\%$ on average, which validates our framework. We also evaluate the performance of the framework for conventional EOS scheduling.

\bigskip


%\end{changemargin}

\normalsize

%\end{frontmatter}

\section{Introduction}
\label{sec:introduction}

Earth observation satellites (EOSs) with specialized cameras are designed to gather images according to user requirements. Broad applications of EOSs can be seen in the fields of Earth resource exploration, environmental monitoring and disaster surveillance, since EOSs provide a large-scale observation coverage. Due to the continuous decrease in launch cost and the improvement in small satellite technology, we are witnessing an explosive growth of the number of orbiting  EOSs and  planned launches, which can be expected to continue in the years to come~\citep{2014Nag}. The scheduling of the EOSs is of great importance for the effective and efficient execution of observation missions~\citep{wang2020agile}.

\blue{EOS scheduling refers to the problem of allocating multiple satellites, each constrained by limited resources such as energy and memory, to a set of observation tasks, with the objective of maximizing the total observation profit over the planning horizon. It involves selecting which targets to observe, at what times, and by which satellites, while respecting physical and operational constraints.}

\blue{Conventional EOS (CEOS) scheduling has been extensively studied~\citep{VasquezHao-255,LinLiao-60,wang2019robust}. In CEOS scheduling, a satellite adjusts only along the roll axis, so its observation time window (OTW) coincides with the visible time window (VTW), as illustrated in Figure~\ref{figEOS}. In contrast, agile EOSs (AEOSs) can maneuver along both the roll and pitch axes, providing additional flexibility to look ahead or look back within the VTW~\citep{LemaitreVerfaillie-269,wang2020agile,song2024generalized}. As shown in Figure~\ref{fig:AEOS-R1}, this enables multiple possible OTWs within a single VTW, significantly enhancing observation capability but also increasing the complexity of scheduling.}

\begin{figure}[tb]
\begin{center}
\includegraphics[width=0.4\textwidth]{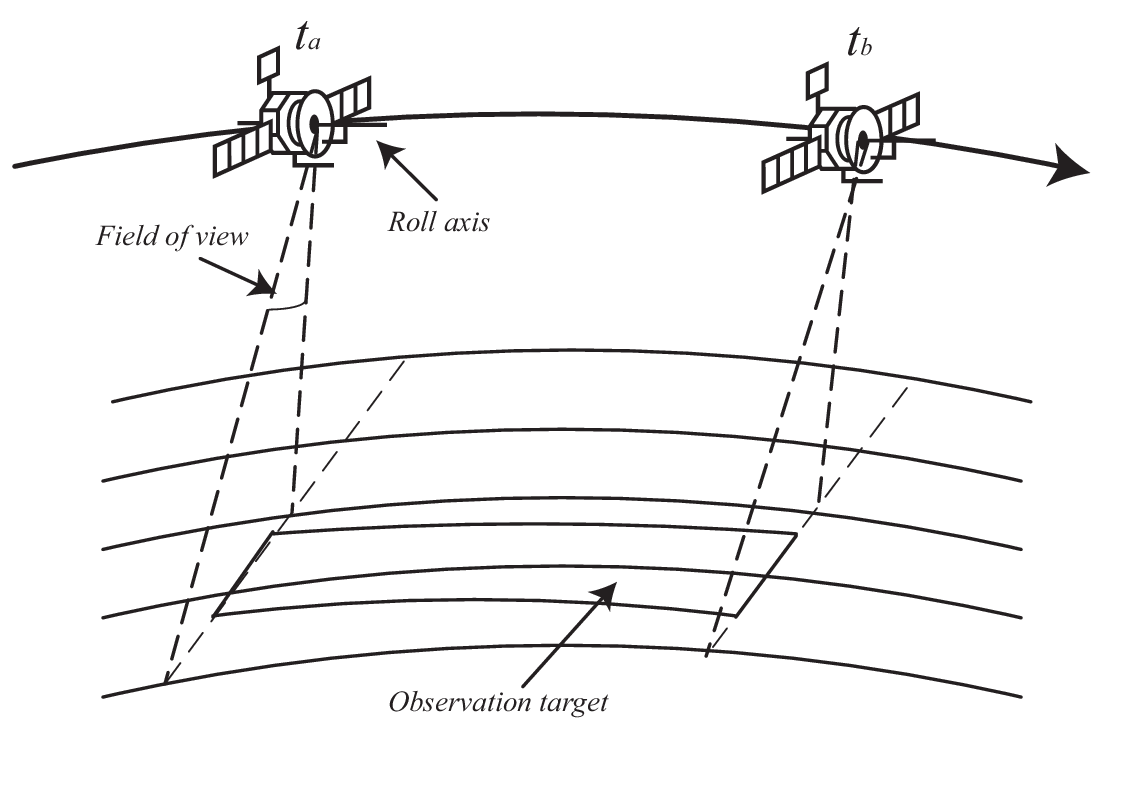}
%\includegraphics[width=2]{Figure1.pdf}
%\space{7cm}
 \caption{\blue{Illustration of the fixed observation interval of a CEOS\@. The CEOS can only maneuver along the roll axis, so that its OTW coincides exactly with the VTW.}}\label{figEOS}
\end{center}
\end{figure}

\begin{figure}[tb]
  \centering
  \includegraphics[width=4.0in]{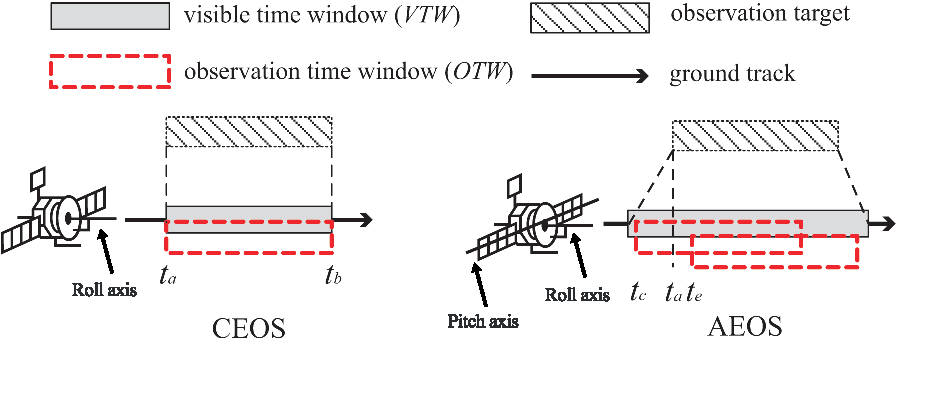}
  \caption{\blue{Comparison of the observation capability of CEOS and AEOS\@. 
Unlike a CEOS, which can only observe within its VTW, an AEOS has additional pitch-axis maneuverability that allows it to look ahead or look back. 
This enables multiple potential OTWs within the same VTW, thereby increasing flexibility and observation opportunities.}}
  \label{fig:AEOS-R1}
\end{figure}

\blue{In this paper, we focus on a variant of the AEOS scheduling problem that allows for multiple observations of the same target, motivated by the need for stereo and time-series imaging~\citep{NicholShaker-105, StearnsHamilton-106,li2024earth}. We define this problem as \emph{multi-observation AEOS scheduling} (MAS)\@. In MAS, each VTW may contain multiple candidate OTWs, and the profit associated with observing a target is a nonlinear function of the number of successful observations. The problem can be regarded as a specific interval scheduling problem in a parallel machine environment, with each orbit of each satellite represented as a machine, and subject to additional constraints such as mission transformation time and on-board memory and energy consumption.}

{The main contributions of this paper are threefold: (1) to the best of our knowledge, scheduling multiple AEOSs with multiple observations has not been studied before; we define the problem MAS with a nonlinear profit function. (2) Since MAS cannot be effectively handled by a general commercial solver in a straightforward manner, we propose a column-generation-based framework based on a reformulation of a compact linear formulation. (3)~We also test the proposed framework for CEOS scheduling and report our findings.}

\blue{The remainder of this paper is organized as follows. Section~\ref{sec:literature} reviews related work on CEOS scheduling and AEOS scheduling. Section~\ref{model} presents the modeling of the MAS problem. Section~\ref{sec:CG} describes the proposed column-generation (CG) framework. Section~\ref{sec:comput} reports computational experiments, and Section~\ref{sec:conclusion} concludes the paper with a discussion of future research directions.}

\section{Literature review}
\label{sec:literature}

\subsection{EOS scheduling}
% In this subsection, we provide a survey of recent work on the related subjects of CEOS and AEOS scheduling. %The models and methods of CEOS scheduling are reported separately.

\cite{GabrelVanderpooten-61} formulate a CEOS scheduling problem  as the selection of a path optimizing multiple criteria in a graph without circuit. Based on the EOS SPOT~5~\citep{dagras1995spot}, \cite{VasquezHao-255} present a formulation as a generalized version of the well-known knapsack model, which includes large numbers of logical constraints. Constraint satisfaction procedures have also been introduced for CEOS scheduling with a set of hard constraints~\citep{Bensana96exactand,Sun2008}. \cite{LinLiao-60} study the daily imaging CEOS scheduling for the EOS ROCSAT-\uppercase\expandafter{\romannumeral2}~\citep{chern2008taiwan}, and develop an integer programming (IP) model with different imaging rewards. For the case where potentially hundreds of orbiting satellites are used to execute observation missions, a window-constrained packing model for CEOS scheduling is established by~\cite{WolfeSorensen-272}, and their work also considers an extended model with multiple resources and multiple observation opportunities for the same target. \cite{wang2019robust} further extend the foregoing models by consideration of uncertainty of cloud coverage and real-time scheduling.

Optimal solutions for CEOS scheduling problems can be found for small instances. \cite{GabrelVanderpooten-61} solve CEOS  scheduling problems with a single satellite using a graph-theoretical procedure. \cite{Bensana96exactand} propose a depth-first branch-and-bound (B\&B) algorithm with  constraint satisfaction ingredients, requiring limited space. Upper bounds for CEOS scheduling problems have been studied in~\cite{Benoist-104} and \cite{VasquezHao-271}. Since exact methods may fail to find optimal solutions in reasonable runtimes, approximate algorithms are viable practical alternatives for CEOS scheduling, especially for large instances with multiple satellites. Various versions of genetic algorithms have been developed for CEOS scheduling \citep{Han-162,KimChang-103,2013Kolici,zhang2022improved}, and heuristics have also been widely applied~\citep{LemaitreVerfaillie-269,2013Kolici,wu2022ensemble}. %Other heuristic methods have also been studied; \cite{WolfeSorensen-272}, for instance, propose a constructive method with fast and simple priority dispatch. \cite{LinLiao-60} describe a Lagrangian relaxation heuristic to obtain near-optimal solutions. \cite{WangReinelt-50} develop a decision support system considering mission conflicts.
\subsection{AEOS scheduling}

The complexity of agile satellite scheduling is significantly higher than for the conventional case, and \blue{a wide range of algorithms have been developed to solve AEOS scheduling problems~\citep{wang2020agile}. These approaches can be broadly grouped into three categories: exact methods, heuristics and metaheuristics, and reinforcement and deep learning methods. In the following, we review representative studies in each category and discuss their applicability to AEOS scheduling.}
\blue{\paragraph{Exact methods.} 
While the literature on exact optimization approaches for AEOS scheduling is limited, some works have explored mathematical programming formulations for small-scale problems~\citep{GabrelMoulet-258,chu2017anytime,peng2020exact}. These methods provide optimality guarantees but are typically restricted to small or simplified instances (e.g., not considering onboard memory and energy constraints) due to exponential computational complexity.}
\blue{\paragraph{Heuristics and metaheuristics.} 
A large body of work has developed heuristic and metaheuristic approaches for AEOS scheduling since the problem was explicitly introduced in the early 21st century. \citet{LemaitreVerfaillie-269} design four simple heuristic algorithms for the AEOS problem. Considering stereoscopic and visibility constraints, \citet{HabetVasquez52} formulate AEOS scheduling as a constrained optimization problem and propose a tabu search method with partial enumeration. \citet{TangpattanakulJozefowiez10} adopt a local search heuristic and compare it with a biased random-key genetic algorithm. \citet{XuChen-5} define priority-based indicators and employ a sequential construction procedure to generate feasible solutions. \citet{WangChen-7} model AEOS scheduling using a directed acyclic graph, and propose a heuristic graph-scanning algorithm. CG-based and simulated annealing heuristics have  also been proposed to handle cloud uncertainty~\citep{wang2021robust,han2022simulated}. Adaptive large neighbourhood search metaheuristics have also been developed, incorporating multiple removal and insertion operators~\citep{liu2017adaptive,he2018improved}.} %Beyond single-satellite problems, limited attention has been paid to multi-AEOS scheduling. For example, \citet{2003Globus} considered integrated scheduling of two AEOSs and transmission operations, evaluating several search techniques including simulated annealing, random swap mutations, and priority ordering. \citet{2007Li} presented a combined genetic algorithm, where the fitness computation was conducted via simulated annealing.}
\blue{More recently, \citet{long2024improved} have proposed an improved particle swarm optimization with reverse learning and neighbor adjustment, while \citet{long2024emergency} develop an event-triggered emergency scheduling method with multi-hierarchical planning.} 
\blue{\paragraph{Learning-based methods.} 
In recent years, reinforcement learning (RL) and deep learning approaches have gained significant traction, offering scalable solutions for large-scale task allocation and complex operational constraints~\citep{miralles2023critical}. For instance, \citet{herrmann2023reinforcement} propose a deep RL framework that integrates Monte Carlo Tree Search and supervised learning to optimize on-board scheduling, achieving improved observation efficiency and robustness. \citet{song2023rl} design a hybrid optimization framework combining Q-learning with a genetic algorithm, while a generalized model and a deep reinforcement learning-based evolutionary method are further developed for multitype satellite observation scheduling \citep{song2024generalized}. \citet{herrmann2024single} demonstrate how deep RL can enable scalable and efficient AEOS scheduling by training a single-agent policy in a simplified environment and deploying it across multiple satellites. More recently, \citet{chen2025conflict} have proposed a two-stage deep network optimization model that balances exploration and exploitation across satellite constellations.}
\blue{\paragraph{Discussion.} 
In summary, exact methods provide strong optimality guarantees but are computationally limited to small problem instances. Heuristic and metaheuristic approaches are widely adopted due to their efficiency and flexibility, but they often lack performance guarantees. Reinforcement and deep learning methods offer scalability and adaptability, but they require significant training data and may face interpretability challenges. Our CG-based framework lies between exact and approximate paradigms, thereby producing high-quality solutions that remain computationally tractable for realistic AEOS instances.}
\blue{It is also worth noting} that the previous models do not allow for more than one observation mission within one $VTW$, and multiple observations for the same target have, as far as we are aware,  not yet been studied for AEOS scheduling. In addition, incorporating multiple observations for multiple agile satellites and producing solutions with a performance guarantee on the optimality gap are still unexplored subjects in the literature; these are exactly the topics that are studied in this article.

\section{Definitions and problem statement}
\label{model}

\blue{Before presenting the detailed mathematical formulation of the MAS problem, we provide an overview of the modeling process. As shown in Figure~\ref{fig:Section3}, the main steps include: problem setup (targets, satellites, orbital data, and VTWs), assumptions and observation rules (e.g., each satellite can observe only one target at a time, no preemption, and resource limits), discretization of VTWs into candidate OTWs, generation of candidate observation missions with attributes such as duration, energy, memory, and profit, and finally the compact MAS modeling framework with decision variables, nonlinear objective function, and capacity constraints.} 
\begin{figure}[!th]
  \centering
  \includegraphics[width=2in]{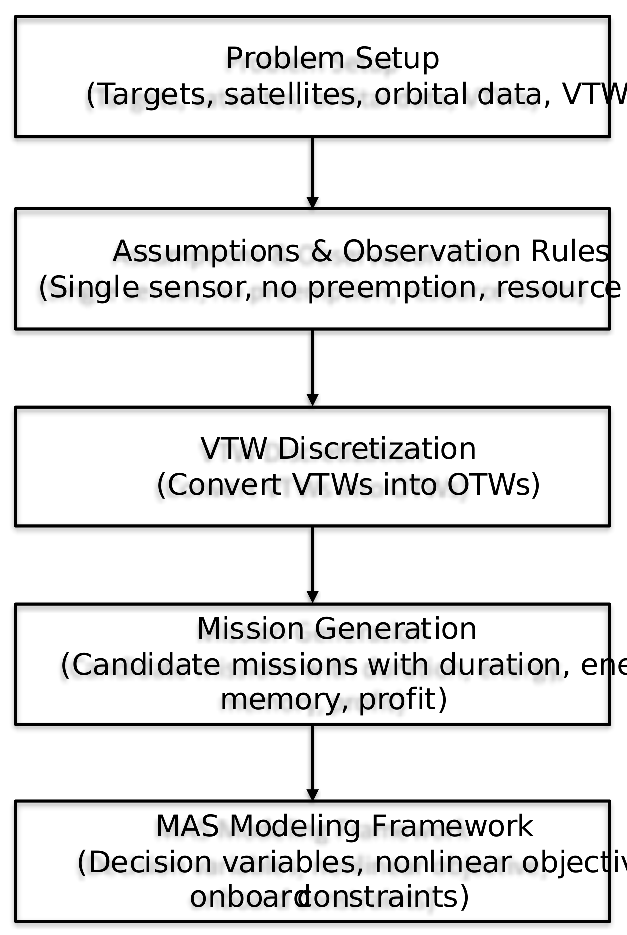}
  \caption{\blue{Overview of the MAS modeling process.
}}
  \label{fig:Section3}
\end{figure}

In Section~\ref{subsec:basicmodel}, we present the basic model MAS, with multiple observations and multiple
agile satellites. \blue{Each satellite is assumed to observe only one target at a time, and observation preemption is not allowed. This is consistent with the operational constraints of Earth observation satellites currently in use. Even when multiple sensors are available, only one can normally be used at a time due to attitude and energy consumption limitations. This assumption has also been widely adopted in the literature~\citep[see, for instance,][]{wang2020agile,chatterjee2024multi,stephenson2025optimal}.}  We will assume that each satellite can only observe one target at a given time, and that observation preemption is not allowed. %Besides, each target is possibly desired to be observed more than once.
Subsequently, in Section~\ref{subsec:extendedmodel}, we extend the basic formulation with practical constraints regarding mission transformation, energy consumption, and memory capacity.

\subsection{MAS} \label{subsec:basicmodel}
Denote $T$ as the set of targets and $\omega_{i}$ as the profit for target $i \in T$.
% As similar profit definition introduced in~\cite{LemaitreVerfaillie-269,Cordeau2005} for the area target, marginal benefit of an extra observation increases for the same target.
In line with the definitions of observation profits of \cite{LemaitreVerfaillie-269} and \cite{Cordeau2005}, we assume that the marginal benefit of an extra observation increases with the number of observations for each target~$i$, up until a maximum desired observation number $N_i$ within the time horizon (although this property is not essential to the model).
%We prefer to schedule enough observation missions for the same target, or abandon partial observations remaining opportunities to observe other targets.
The observation profit~$\omega_{i}$ for target $i$ is a nonlinear function of the observation count, which can be written in linearized form as $\omega_{i} = \sum_{s=0}^{s=N_{i}} \pi_{is}  y_{is} $ with $\sum_{s=0}^{s=N_{i}} y_{is} = 1$, where all $y_{is}$ are binary variables ($i\in T$, $s=0,\ldots,N_i$).  When
 $s$~observations are scheduled for target $i$ then
 $\pi_{is}$ is the observation profit (a nonnegative integer) and % and $N_{i}$ is the user-desired observation number within time horizon.
$y_{is}=1$; otherwise $y_{is}=0$. We let $y_{iN_{i}}=1$ when the number of missions for target $i$ is greater than $N_{i}$, since no additional profit would be received. An illustration of the profit function of a target $i$ is provided in Figure~\ref{PieceProFun}, with the maximum number of missions $N_{i} = 4$, and $\pi_{i0}=0$, $\pi_{i1}=1$, $\pi_{i2}=3$, $\pi_{i3}=6$, and $\pi_{i4}=10$.

%\begin{align}
%&\omega_{i}(l) = \sum_{s=0}^{s=N_{i}} \pi_{is}  y_{is} & \forall i \in T\\
%\text{with} \quad &\sum_{s=0}^{s=N_{i}} y_{is} = 1& \forall i \in T
%\end{align}

\begin{figure}[!t]
\begin{center}
\begin{tikzpicture}
\begin{axis}[
xlabel={Number of observations for target $i$},ylabel={Observation profit of target $i$},
xmin=-0.1,xmax=4.3,ymin=-0.5,ymax=11,
ytick={0,1,3,6,10}
]
\draw[dotted,blue,line width=0.5 mm] (0,0)--(1,0)--(1,1)--(2,1)--(2,3)--(3,3)--(3,6)--(4,6)--(4,10);
\draw[red,fill=red] (0,0) circle (0.5ex);
\draw[red,fill=red] (1,1) circle (0.5ex);
\draw[red,fill=red] (2,3) circle (0.5ex);
\draw[red,fill=red] (3,6) circle (0.5ex);
\draw[red,fill=red] (4,10) circle (0.5ex);
\end{axis}
\end{tikzpicture}
\caption{\blue{Illustration of the profit function for a single target $i$ with up to four observations ($N_i=4$). The profit values ($\pi_{i0}=0$, $\pi_{i1}=1$, $\pi_{i2}=3$, $\pi_{i3}=6$, $\pi_{i4}=10$) are illustrative examples to demonstrate the nonlinear increase.}}
\label{PieceProFun}
\end{center}
\end{figure}

Define $S$ as the set of satellites, and %. Since the AEOS can observe targets in different orbits,
let $B_{j}$ represent the set of orbits for satellite $j \in~S$. The set of candidate observation missions in orbit $k \in B_{j}$ is defined as $M_{jk}$. Each candidate observation mission $p\in M_{jk}$ is expressed as a pair $(OTW_{jkp}, i)$, indicating that each observation mission is associated with a corresponding target $i$ and a specific $OTW$.

We will consider sequence-dependent constraints concerning mission transformation and energy in Section~\ref{subsec:extendedmodel}. To ensure that these constraints can be easily incorporated into the model, the binary decision variable $x_{jkpq}$ is adopted, with $x_{jkpq}=1$ when observation mission~$q$ is the immediate successor of $p$ in orbit $k\in B_{j}$ and $x_{jkpq}=0$ otherwise. For each satellite orbit, we add dummy missions $s_{jk}$ and $e_{jk}$ as source and sink node, respectively. The number of observations scheduled for target $i$ is then expressed as $\sum\limits_{j\in{S}}\sum\limits_{k\in{B_{j}}}\sum\limits_{p\in{M_{jk}^{i}}} \sum\limits_{q\in{M_{jk}\cup{\{e_{jk}\}}}} x_{jkpq}$, where $M_{jk}^{i}$ stands for the observation mission set associated with target $i$ in orbit $k\in B_{j}$. A compact linear formulation of MAS can then be stated as follows:
\begin{align}
\text{max}\quad&  \sum\limits_{i\in{T}} \sum_{s=1}^{s=N_{i}} \pi_{is} y_{is}                                               \label{NewObjFunc}\\
  \text{subject to} &\sum\limits_{q\in{M_{jk\cup e}}} x_{jkpq} - \sum\limits_{q\in{M_{jk\cup s}}} x_{jkqp} =\begin{cases}
             1, &p=s_{jk}    \\
             0, &\forall{p\in{M_{jk}}}\\
         -1,& p=e_{jk}
             \end{cases} & \forall k\in{B_{j}},j\in{S}
              \label{FlowCon0}  \\
 &  \sum\limits_{j\in{S}} \sum\limits_{k\in{B_{j}}} \sum\limits_{{p\in{M_{jk}^{i}}}}\sum\limits_{q\in{M_{jk\cup e}}} x_{jkpq} \geq \sum_{s=0}^{s=N_{i}} s \cdot y_{is} & \forall{i}\in{T}
 \label{OSDYCons}  \\
 &  \sum_{s=0}^{s=N_{i}} y_{is} = 1  & \forall {i}\in{T}    \label{YCons0}
\end{align}
where $M_{jk\cup s} = M_{jk} \cup \{ s_{jk}\}$ and $M_{jk\cup e} = M_{jk} \cup \{ e_{jk}\}$.

The objective function~\eqref{NewObjFunc} aims to maximize the observation profit for the targets. \blue{Note that when $s=0$, the corresponding profit is zero ($\pi_{i0} = 0$). For this reason, the summation of the observation profits in~\eqref{NewObjFunc} can start from $s=1$, as the term $s=0$  does not contribute to the objective value.} The constraints~\eqref{FlowCon0} represent flow generation and conservation. Constraint set~\eqref{OSDYCons} guarantees that the profits are computed correctly according to the number of scheduled observations, and constraints~\eqref{YCons0} ensure that each target receives exactly one profit value.

\subsection{Extensions} \label{subsec:extendedmodel}
%In this section, we extend the basic formulation considering mission transformation, energy consumption and memory capacity.
%The data transmission is not considered since we assume there are enough ground transmission stations and relay satellites to download data in each satellite orbit.

The transformation time $\Delta_{jkpq}^{T}$ between two observation missions $p$ and $q$ in orbit $k\in B_{j}$ consists of attitude maneuvering time $\Delta_{jkpq}^{V}$ and attitude stabilization time $\Delta_{jkpq}^{S}$, which are given as inputs~\citep{wertz1978spacecraft}. Since the orbital duration of an AEOS is much larger than the transformation time between missions, the transformation constraint for two observation missions in different orbits is always satisfied, and is not imposed separately.  This allows to partition each satellite's time horizon into independent orbits, which greatly reduces the number of constraints. The transformation constraints are checked in a preprocessing stage: decision variable $x_{jkpq}$ is defined only if the sum of the completion time of mission $p$ and $\Delta_{jkpq}^{T}$ is not greater than the starting time of mission~$q$ in orbit $k\in B_{j}$.

The memory capacity in one orbit of satellite $j$ is defined as $M^{C}_j$, and the unit-time imaging memory occupation is denoted as $M^{I}_j$. The memory capacity constraints are formulated per orbit, since we assume that satellites can transfer data to the ground station after each orbit. The energy system of an AEOS is typically partially supported by a solar panel collecting energy from the Sun. Although the  conditions for solar energy collection are variable due to environmental variation, the amount of energy collection in one orbit is nearly constant~\citep{WangDemeulemeester-4}. We therefore assume that the maximal energy capacity of satellite $j$ is constant and is denoted as $E^{C}_j$ in each orbit. Unit-time imaging energy consumption and maneuvering energy consumption are defined as $E^{I}_j$ and $E^{M}_j$ for satellite $j$, respectively. The number of observation missions in each orbit is limited to satisfy memory and energy constraints.

The extended formulation for MAS is then to maximize~\eqref{NewObjFunc} subject to~\eqref{FlowCon0}--\eqref{YCons0} and
 %\begingroup \makeatletter \def\f@size{8}\check@mathfonts
 \begin{align}
&\sum\limits_{{p\in{M_{jk}}}}\sum\limits_{q\in{M_{jk\cup e}}} x_{jkpq} d_{jkp} M^{I}_{j} \leq{M^{C}_j} & \forall{k}\in{B_{j}}, {j}\in{S} % Memory Storage Constraints
\label{NewMemoryCon}
\\
&\sum\limits_{{p\in{M_{jk}}}}\sum\limits_{q\in{M_{jk\cup e}}} x_{jkpq}d_{jkp} E^{I}_{j} +\sum\limits_{{p\in{M_{jk}}}}\sum\limits_{{q\in{M_{jk}}}} x_{jkpq}\Delta_{jkpq}^{V} E^{M}_{j} \leq{E^{C}_j} &\forall{k}\in{B_{j}},{j}\in{S}
\label{NewEnergyCon}
 \end{align}
% \endgroup
\noindent where $d_{jkp}$ is the observation duration of $OTW_{jkp}$. % Clearly, these knapsack-type constraints imply that the problem is NP-hard.

\blue{The model explicitly accounts for satellite memory and energy capacities, and includes sequence-dependent transition times (maneuver and stabilization). Other operational factors, such as cloud uncertainty and telemetry, tracking, and command constraints, are not considered in this study and are left for future work.} 
\blue{With the inclusion of the resource constraints (memory and energy), the problem MAS is equivalent to a multi-dimensional knapsack problem and is therefore clearly NP-hard.}

\section{Column generation} \label{sec:CG}

CG is a promising method to tackle formulations with a huge number of variables~\citep{Bar1998,Wilhelm2001,GSCHWIND2018521}. This technique has been used for decades since CG was first applied to the cutting stock problem as part of an efficient heuristic algorithm~\citep{1961Gilmore,1963Gilmore}. The main advantage of CG is that not all variables need to be explicitly included into the model. In this section, we decompose the linear relaxation of the compact formulation of Section~\ref{model}, leading to a CG-based solution framework for the LP-relaxation.  We design a column initialization heuristic, and iteratively solve a restricted master problem (RMP) with restricted column set by a commercial LP-solver. We apply a label-setting method to find solutions to the pricing problems in the CG-procedure, thus iteratively adding new columns until an optimal solution of the RMP is found. The integrality constraints are then re-imposed using all the generated columns to obtain a high-quality solution with an IP-solver. In light of the large number of variables for MAS in practice, in the corresponding B\&B routine we do not generate new columns for the new LP-problems encountered upon branching because this would become overly time-consuming. In other words, we apply a CG-heuristic and we do not implement a branch-and-price procedure. Similar choices are frequently made in the literature, for instance in~\citet{FURINI2012251} and \citet{GUEDES2015361}. We will show in Section~\ref{sec:comput} that the tight LP-bound confirms that a near-optimal integer solution is usually obtained in this way.

\blue{In the introduced CG-based framework, the RMP is a linear program, solvable in polynomial time as a function of the instance size (provided that the number of columns does not grow exponentially). The pricing subproblems to be introduced later are resource-constrained shortest path problems, for which the worst-case complexity is exponential with the number of vertices in the graph, knowing that the resource-constrained shortest path problem is NP-hard even with one resource \citep{MehlhornZiegelmann}. In practice, however, the number of dominant paths may be relatively low thanks to the resource constraints and dominance criteria \citep{GabrelVanderpooten-61}, and the entire algorithm empirically turns out to be a computationally viable alternative, as demonstrated in Section~\ref{sec:comput}  through computational experiments.}

\subsection{Dantzig-Wolfe decomposition}

The flow-based formulation of Section~\ref{model} has been used to design heuristic and constructive algorithms~\citep{LinLiao-60,liu2017adaptive}, but it is difficult to evaluate the optimality gap for such procedures. Exact methods, on the other hand, fail to obtain optimal solutions for large instances~\citep{GabrelVanderpooten-61}. Since the constraints are decoupled into different orbits, we apply Dantzig-Wolfe reformulation \citep{martin2012large} to decompose the original model.
\newtheorem{defn}{Definition}
\newtheorem{prop}{Property}

The schedules in each satellite orbit are regarded as columns. Denote the set of schedules in satellite orbit $k\in B_{j}$ as $R_{jk}$. Each schedule $m\in R_{jk}$ contains values $x_{jkpq}^{m}$, where $x_{jkpq}^{m}=1$ if mission $p$ is the immediate predecessor of mission $q$ according to schedule $m$ in orbit $k\in \blue{B}_{j}$; $x_{jkpq}^{m}=0$ otherwise. We introduce a binary variable $z_{jkm}$ for each $m\in R_{jk}$ such that $z_{jkm}=1$ when schedule $m$ is chosen and 0 otherwise. The master problem then aims to maximize~\eqref{NewObjFunc} subject to~\eqref{YCons0} and
 \begin{align}
&\sum\limits_{j\in{S}}\sum\limits_{k\in{B_{j}}} \sum\limits_{m\in{R_{jk}}}\sum\limits_{{p\in{M_{jk}^{i}}}}\sum\limits_{q\in{M_{jk\cup e}}} x_{jkpq}^{m}  z_{jkm} \geq \sum_{s=1}^{s=N_{i}} s \cdot y_{is} & \forall{i}\in{T}
\label{SDYCons2}
\\
&\sum\limits_{m\in{R_{jk}}} z_{jkm} = 1& \forall k\in{B_{j}},{j}\in{S}
\label{ZCons}
\end{align}
\noindent where constraints~\eqref{SDYCons2} correspond with~\eqref{OSDYCons} in the original flow formulation and constraints~\eqref{ZCons} ensure that only one orbit schedule is selected for each satellite orbit.
For each orbit, new columns are iteratively generated according to the results of the pricing problems detailed in Section~\ref{subsec:pricing}. % Eventually we will use the existing columns to solve the master problem until there is no new column generated by the pricing.

\subsection{The pricing problems} \label{subsec:pricing}

We remove the integrality constraints to obtain the LP relaxation of the master problem. With dual variables ${\theta}_{i}^{1}$ associated with constraints~\eqref{SDYCons2}, variables ${\theta}_{i}^{2}$ with constraints~\eqref{YCons0}, and variables~${\theta}_{jk}^{3}$ with constraints~\eqref{ZCons}, we have the following dual:
\begin{align}
\text{min}\quad & \sum\limits_{i\in{T}}{\theta}^{2}_{i} +  \sum\limits_{j\in{S}}\sum\limits_{k\in{B_{j}}} {\theta}^{3}_{jk}
\label{DualObjFunc}\\
\text{s.t.}&-s  {\theta}^{1}_{i} + {\theta}^{2}_{i}  \geq \pi_{is} & \forall{s}\in\{1, 2, \ldots, N_{i}\}, {i}\in{T}
\label{DualC1}\\
&\sum\limits_{i\in{T}}{\theta}^{1}_{i} \sum\limits_{{p\in{M_{jk}^{i}}}}\sum\limits_{q\in{M_{jk\cup e}}} x_{jkpq}^{m} +  {\theta}^{3}_{jk} \geq 0& \forall m \in {R_{jk}},k\in{B_{j}},{j}\in{S}
\label{DualC2}\\
&{\theta}^{1}_{i}\leq 0,{\theta}^{2}_{i} \in \mathbb{R}& \forall{{i}\in{T}}
\label{DualV1}\\
&{\theta}^{3}_{jk}\in \mathbb{R}& \forall k\in{B_{j}},{j}\in{S}
\label{DualV3}
\end{align}

At each CG-iteration we check the violation of constraints~\eqref{DualC2}. The LP relaxation is typically solved faster if we add constraints that are strongly violated, and we will search for columns with most negative reduced cost. Considering the special structure of the dual formulation, %the corresponding pricing problem is decoupled into several pricing problems according to the satellite orbits.
we end up with a separate pricing problem for each satellite orbit $k\in B_{j}$, as follows:

\begin{align}
\text{min}\quad & \sum\limits_{i\in{T}}{\theta}^{1}_{i} \sum\limits_{p\in{M_{jk}^{i}}} \sum\limits_{q\in{M_{jk\cup e}}}  x_{jkpq}  +  {\theta}^{3}_{jk}
\label{PriObjFunc} \\
\text{s.t.}&\sum\limits_{q\in{M_{jk\cup e}}} x_{jkpq} - \sum\limits_{q\in{M_{jk\cup s}}} x_{jkqp} =
\begin{cases}
             1,& p=s_{jk}    \\
             0,& \forall{p\in{M_{jk}}}\\
         -1,& p=e_{jk}
             \end{cases}
              %Observation Frequency Constraints
\label{FlowConCG}
\\
&\sum\limits_{{p\in{M_{jk}}}}\sum\limits_{q\in{M_{jk\cup e}}} x_{jkpq} d_{jkp} M^{I}_{j} \leq{M^{C}_j}% Memory Storage Constraints
\label{NewMemoryConCG}
\\
&\sum\limits_{{p\in{M_{jk}}}}\sum\limits_{q\in{M_{jk\cup e}}} x_{jkpq} d_{jkp} E^{I}_{j} + \sum\limits_{{p\in{M_{jk}}}}\sum\limits_{{q\in{M_{jk}}}} x_{jkpq} \Delta_{jkpq}^{V} E^{M}_{j} \leq{E^{C}_j} % Energy Constraints
\label{NewEnergyConCG}
\end{align}

For each pricing problem, if the optimal value is less than 0 then a new column for the corresponding orbit is generated with lowest reduced cost; otherwise no new column results. The CG-iterations continue until all pricing problems return an optimal value not less than 0, demonstrating that there are no violated constraints~\eqref{DualC2}.  The corresponding LP objective value constitutes a tight upper bound for problem \blue{MAS}.% The tight property of the upper bound is revealed in the computational experiments.

\subsection{Column initialization}

Our column initialization heuristic proceeds as follows. For each satellite orbit, the algorithm attempts to generate initial columns for a given number of iterations. In each iteration, the observation missions are re-ordered randomly and it is checked whether they can be greedily added into the current schedule. Since different orbit schedules may have different contributions to the overall profit, we introduce the following definition. % of column dominance as well as the related property to further improve the effectiveness of column initialization.

%\begin{defn} Dominant observation strategy. An observation strategy is named dominant if we cannot add more observation missions into the observation strategy while considering constraints.  \end{defn}

\begin{defn} Column dominance. For a satellite orbit, column dominance occurs when the number of scheduled observations for any target in one (dominant) column is no less than the observation number of the corresponding target in another (dominated) column, and for at least one target, the observation number of the dominant column is strictly greater than in the dominated column.
%the corresponding number in the other column.
\end{defn}

\begin{prop} If one column is dominated by the existing columns, it can be discarded without %missing the optimal value of
deteriorating the objective function.\end{prop}

%\begin{prop} Several dominant observation strategies can result in the same scheduling scheme.\end{prop}
The dominated column can be discarded since the dominant column has a higher contribution to objective. A generated column is discarded if it is dominated; otherwise it  is added to the initial column pool. Similarly, %We also check whether current column dominates the existing columns in $R_{jk}$.
if a column currently in the column pool is dominated by a new column then it is also removed.

%\begin{algorithm}[ht]
%\caption{Column initialization heuristic}
%\label{Algorithm1}
%\begin{algorithmic}[1]
%\REQUIRE ~~\\ %算法的输入参数：Input
%Candidate observation missions $\bigcup \limits_{k\in B_{j},j\in S} M_{jk}$ and number of iterations for each satellite orbit $nbI$;\\
%\ENSURE ~~\\ %算法的输出：Output
%Set of orbit schedules $R_{jk}$;\\
%\FOR{each $j \in S$}
%\FOR{each $k \in B_{j}$}
%\STATE $R_{jk} = \emptyset$; \
%\FOR{$1$ to $nbI$}
% \STATE$CurR= \emptyset$;\
%\STATE Randomly reorder the candidate missions of set $M_{jk}$; \
%\FOR{each $p \in M_{jk}$}
%\IF{$CheckCons(CurR, p)$}
%\STATE Update $CurR$; \
%\ENDIF
%\ENDFOR
%\IF{$CheckDominance(R_{jk},CurR)$}
%\STATE  $R_{jk}= R_{jk} \cup \{CurR\}$; \
%\ENDIF
%\ENDFOR
%\ENDFOR
%\ENDFOR
%
%\end{algorithmic}
%\end{algorithm}

\subsection{Label-setting method for pricing}% problems}

The pricing problem for each satellite orbit %require to be solved iteratively. Besides, we observe that the pricing problem,
%searching the column with the most negative reduced cost,
corresponds to a resource-constrained shortest path problem~\citep{pugliese2013survey} in a graph without circuit.  We use an adaptation of a label-setting algorithm~\citep{GabrelVanderpooten-61} to solve this problem.
%Although the number of paths grows exponentially with the number of missions, the candidate observation missions in each orbit are far from being exponential in practice and the number of useful paths can be relatively low.
For each satellite orbit $k \in B_{j}$, we define a directed acyclic graph $G_{jk}=(N_{jk},E_{jk})$, where the vertex set $N_{jk}$ contains all candidate observation missions together with dummy missions $s_{jk}$ and $e_{jk}$.  The edge set $E_{jk}$ contains an edge from mission (node) $p$ to mission (node) $q$ only if $q$ can be executed immediately after $p$; we also include an edge from $s_{jk}$ to each other vertex, and $e_{jk}$ has edges incoming from all other vertices.
%Then the mission path is defined in graph $G_{jk}$ as follows.

\begin{defn} Mission path. In the directed acyclic graph $G_{jk}=(N_{jk},E_{jk})$, a mission path is a path from the dummy source to any observation mission or the dummy sink mission $q \in N_{jk}\setminus\{s_{jk}\}$.  A mission path can be represented as a tuple $P_{jkqt}=(Cost_{jkqt},CurM_{jkqt},CurE_{jkqt})$, where $Cost_{jkqt}$ is  the sum of the mission costs along the path, $CurM_{jkqt}$ and $CurE_{jkqt}$ are the summed memory occupation and energy consumption, respectively, and value $t$ is an index for the mission paths ending at mission $q$ in $G_{jk}$.
\end{defn}

For each observation mission $q$ in graph $G_{jk}$, the mission cost $MisCost_{jkq}$ equals $\theta^{1}_{i}$, where $i$ is the target associated with mission $q$ (see objective function \eqref{PriObjFunc}). The mission memory occupation and energy consumption are denoted as $MisM_{jkq}$ and $MisE_{jkq}$. For the dummy source and sink missions, the mission cost is set as $\theta^{3}_{jk}/2$, and the mission memory and energy consumption are set to 0. Searching a column with the most negative reduced cost in orbit $k\in B_{j}$ now transforms to searching a constrained path with minimal mission cost from source to sink in $G_{jk}$. The resource constraints ensure that the total memory occupation and energy consumption along each path do not exceed the corresponding capacity.

Denote the set of mission paths ending at mission $q$ in $G_{jk}$ as $\mathscr{P}_{jkq}$. Property~\ref{prop:pathdom} below enhances algorithmic efficiency of the label-setting method used to find a minimal-cost path; the details of the method are described in Appendix~\ref{ApenB}.
%;  of mission path along with the related property is then defined to ensure the algorithm efficiency.

\begin{defn} Mission path dominance. In the set $\mathscr{P}_{jkq}$ for a given mission $q$ in $G_{jk}$, mission path $P_{jkq{t_{1}}}$ dominates $P_{jkq{t_{2}}}$ when the following inequalities hold and at least one of the inequalities is strict: 1) $Cost_{jkqt_{1}} \leq Cost_{jkqt_{2}}$; 2) $CurM_{jkqt_{1}} \leq CurM_{jkqt_{2}}$; 3) $CurE_{jkqt_{1}} \leq CurE_{jkqt_{2}}$. \end{defn}

\begin{prop} \label{prop:pathdom} For the pricing problem in a given satellite orbit, a dominated mission path can be discarded without loss of optimality.  \end{prop}

\begin{proof}
Without loss of generality, we denote the dominant and dominated path as $P_{jkq{t_{1}}}$ and $P_{jkq{t_{2}}}$ in $\mathscr{P}_{jkq}$, respectively. Assume that there exists a shortest path $P_{jke_{jk}t}$ in $G_{jk}$ including path $P_{jkq{t_{2}}}$. If we replace $P_{jkq{t_{2}}}$ by $P_{jkq{t_{1}}}$  while maintaining the subsequent path from $q$ to $e_{jk}$, %. According to the definition of mission path dominance,
we obtain a new feasible path $P_{jke_{jk}t^{\prime}}$ satisfying all resource constraints and with $Cost_{jke_{jk}t^{\prime}} \leq Cost_{jke_{jk}t}$. %The optimality is ensured.
\end{proof}

\section{Computational experiments} \label{sec:comput}

\subsection{Data generation}
Since there is no common benchmark for AEOS scheduling~\citep{liu2017adaptive}, the proposed CG-based framework is tested on a diverse set of realistic instances that is generated as follows. Our satellite configuration is based on China's high-resolution AEOSs \textit{Gaojing-1} (also known as \mbox{\textit{SuperView-1}}) \citep{li2021china}. \textit{Gaojing-1} is a commercial constellation of four remote sensing satellites. Details of the satellites' parameters are provided in Appendix~\ref{ApenA}.

Following~\cite{liu2017adaptive}, the observation targets are randomly distributed worldwide, and in specific interest areas. We consider these two target distributions together in the same instance within 24 hours, with 150 globally distributed targets and several specific interest areas with 50~targets. The number of interest areas varies from 0 to 3, and thus the total number of observation targets is 150, 200, 250, or 300. The desired maximum observation number $N_i$ for each target $i$ is randomly generated from 1 to 5 and the maximal profit $\pi_{iN_{i}}$ for each target is defined as an integer number from 1 to 10. Since the satellites in the \textit{Gaojing-1} constellation have very similar properties, the constraints for each satellite are taken to be identical. The unit-time imaging memory occupation $M^{I}$ is 10 MB per second, while the unit-time imaging \blue{energy} consumption $E^{I}$ and maneuvering energy consumption $E^{M}$ are 500 and 1000 Watt, respectively. The memory capacity~$M^{C}$ is set as 400, 500, or 600 MB\@. The energy capacity $E^{C}$ is considered in electric charge from 30, 40, 50, or 60 kilojoules (kJ). The discretization unit of each $VTW$ is set as two seconds (see Section~\ref{subsubsec:discretization} for more information). With the above settings, the length of each $VTW$ is around 90 seconds, generating about 45 $OTW$s. For each combination of parameter settings, 10 instances are randomly generated.

\subsection{Computational results}

\subsubsection{Experimental setup}
The computational experiments are conducted on a laptop equipped with Intel Core i5-7200 CPU at 2.5 GHz and 8 GB of RAM on a Windows 10 64-bit OS\@. The algorithm is implemented in Visual C++. The LP- and IP-solver is CPLEX 12.6.3 using Concert Technology with four threads on two cores. The time limit for each run is set as 1200 seconds. In the tables below, the columns labelled $opt$ and $time$ contain the number of instances solved to guaranteed optimality out of 10 instances per setting and the average CPU time for the 10 instances in seconds, respectively. Columns $ub$ show the number of instances for which the LP is successfully solved to optimality, providing the upper bound. The entries labelled $gap$ represent the relative gap between the scheduling solution and the LP-bound obtained from the CG-based framework. \blue{This value is computed as ${gap} = \frac{z_{\text{LP}} - z_{\text{CGH}}}{z_{\text{LP}}} \times 100\%$, where $z_{\text{LP}}$ denotes the objective value of the LP-relaxation and $z_{\text{CGH}}$ is the objective value of the integer solution obtained by the CG heuristic.}

\subsubsection{Results for AEOS instances}
We first look into the results of the compact flow formulation solved by CPLEX, which are summarized in Table~\ref{tab:GlobalFlow}. Clearly, the flow formulation  already struggles with the smallest instances, especially when the memory capacity increases. For the instances with 150 globally distributed targets, only 24 out of 120 instances are solved to optimality, and its overall runtime is orders of magnitude higher than for the CG-based framework reported below. Due to this poor computational performance, we will not further include this formulation for comparison on larger instances.

% Table generated by Excel2LaTeX from sheet 'Sheet4'
\begin{table}[!b]
  \centering
  \footnotesize
  \caption{Results of the flow formulation on instances with 150 targets.}
    \begin{tabular}{ccrr}
     \toprule
    $M^{C}$ \blue{(MB)}  & $E^{C}$ \blue{(kJ)} & $opt$ $( /10)$ & $time$ \blue{(s)} \\
    \midrule
    400   & 30   & 0     & 1200.00 \\
          & 40   & 0     & 1200.00 \\
          & 50   & 9     & 566.09 \\
          & 60   & 10    & 255.93 \\
    500   & 30   & 0     & 1200.00 \\
          & 40   & 0     & 1200.00 \\
          & 50   & 0     & 1200.00 \\
          & 60   & 5     & 797.94 \\
    600   & 30   & 0     & 1200.00 \\
          & 40   & 0     & 1200.00 \\
          & 50   & 0     & 1200.00 \\
          & 60   & 0     & 1200.00 \\
           \midrule
    \multicolumn{2}{c}{Overall} & 24    & 1035.00 \\
     \bottomrule
    \end{tabular}%
  \label{tab:GlobalFlow}%
\end{table}%

Next, we report results for the CG-based framework. We denote our algorithm with label-setting pricing solver as CG-LAB, and we compare with CG-CPL in which we call CPLEX to solve the pricing problems. The outcomes are gathered in Tables \ref{tab:TableAgile150} to \ref{tab:TableAgile300}, for instances with number of targets ranging from 150 to 300.

\begin{table}[p]
  \centering
    \footnotesize
  \caption{Computational results of CG-LAB and CG-CPL on instances with 150 targets.}
    \begin{tabular}{ccrrrrrrr}
    \toprule
          &       & \multicolumn{2}{c}{CG-CPL} &       & \multicolumn{2}{c}{CG-LAB} &       &  \\
\cmidrule{3-4}\cmidrule{6-7}    $M^{C}$ \blue{(MB)} & $E^{C}$ \blue{(kJ)}& \multicolumn{1}{c}{$ub$ $( /10)$} & \multicolumn{1}{c}{$time$}\blue{(s)} &       & \multicolumn{1}{c}{$ub$ $( /10)$} & \multicolumn{1}{c}{$time$}\blue{(s)} &       & \multicolumn{1}{c}{$gap$ $(\%)$} \\
 \midrule
    400   & 30   & 10    & 760.24 &       & 10    & 36.70 &       & 2.56 \\
          & 40   & 6     & 966.23 &       & 10    & 37.20 &       & 1.60 \\
          & 50   & 10    & 129.14 &       & 10    & 20.21 &       & 0.84 \\
          & 60   & 10    & 151.98 &       & 10    & 11.29 &       & 0.90 \\
    500   & 30   & 10    & 864.85 &       & 10    & 36.70 &       & 2.56 \\
          & 40   & 3     & 1155.61 &       & 10    & 37.20 &       & 1.60 \\
          & 50   & 2     & 1183.87 &       & 10    & 20.21 &       & 0.84 \\
          & 60   & 10    & 206.12 &       & 10    & 11.29 &       & 0.90 \\
    600   & 30   & 10    & 631.21 &       & 10    & 36.70 &       & 2.56 \\
          & 40   & 3     & 1179.33 &       & 10    & 37.20 &       & 1.60 \\
          & 50   & 0     & 1200.00 &       & 10    & 20.21 &       & 0.84 \\
          & 60   & 0     & 1200.00 &       & 10    & 11.29 &       & 0.90 \\
    \midrule
    \multicolumn{2}{c}{Overall} & 74    & 802.38 &       & 120   & 26.35 &       & 1.48 \\
    \bottomrule
    \end{tabular}%
  \label{tab:TableAgile150}%
\end{table}%

% Table generated by Excel2LaTeX from sheet 'Sheet2'
\begin{table}[p]
  \centering
    \footnotesize
  \caption{Computational results of CG-LAB and CG-CPL on instances with 200 targets.}
    \begin{tabular}{ccrrrrrrr}
    \toprule
          &       & \multicolumn{2}{c}{CG-CPL} &       & \multicolumn{2}{c}{CG-LAB} &       &  \\
\cmidrule{3-4}\cmidrule{6-7}   $M^{C}$ \blue{(MB)} & $E^{C}$ \blue{(kJ)}& \multicolumn{1}{c}{$ub$ $( /10)$} & \multicolumn{1}{c}{$time$}\blue{(s)} &       & \multicolumn{1}{c}{$ub$ $( /10)$} & \multicolumn{1}{c}{$time$}\blue{(s)} &       & \multicolumn{1}{c}{$gap$ $(\%)$} \\
 \midrule
    400   & 30   & 2     & 1187.84 &       & 10    & 78.75 &       & 3.23 \\
          & 40   & 0     & 1200.00 &       & 10    & 81.95 &       & 2.26 \\
          & 50   & 10    & 206.48 &       & 10    & 43.03 &       & 0.98 \\
          & 60   & 10    & 182.45 &       & 10    & 26.32 &       & 0.92 \\
    500   & 30   & 0     & 1200.00 &       & 10    & 74.92 &       & 3.17 \\
          & 40   & 0     & 1200.00 &       & 10    & 93.60 &       & 2.64 \\
          & 50   & 0     & 1200.00 &       & 10    & 109.18 &       & 2.40 \\
          & 60   & 10    & 232.87 &       & 10    & 75.69 &       & 1.81 \\
    600   & 30   & 0     & 1200.00 &       & 10    & 74.43 &       & 3.23 \\
          & 40   & 0     & 1200.00 &       & 10    & 95.02 &       & 2.59 \\
          & 50   & 0     & 1200.00 &       & 10    & 118.69 &       & 2.62 \\
          & 60   & 0     & 1200.00 &       & 10    & 140.14 &       & 2.63 \\
    \midrule
    \multicolumn{2}{c}{Overall} & 32    & 950.80 &       & 120   & 84.31 &       & 2.37 \\
    \bottomrule
    \end{tabular}%
  \label{tab:TableAgile200}%
\end{table}%

% Table generated by Excel2LaTeX from sheet 'Sheet2'
\begin{table}[p]
  \centering
    \footnotesize
  \caption{Computational results of CG-LAB and CG-CPL on instances with 250 targets.}
    \begin{tabular}{ccrrrrrrr}
    \toprule
          &       & \multicolumn{2}{c}{CG-CPL} &       & \multicolumn{2}{c}{CG-LAB} &       &  \\
\cmidrule{3-4}\cmidrule{6-7}    $M^{C}$ \blue{(MB)} & $E^{C}$ \blue{(kJ)}& \multicolumn{1}{c}{$ub$ $( /10)$} & \multicolumn{1}{c}{$time$}\blue{(s)} &       & \multicolumn{1}{c}{$ub$ $( /10)$} & \multicolumn{1}{c}{$time$}\blue{(s)} &       & \multicolumn{1}{c}{$gap$ $(\%)$} \\
 \midrule
    400   & 30   & 0     & 1200.00 &       & 10    & 167.03 &       & 4.47 \\
          & 40   & 0     & 1200.00 &       & 10    & 242.30 &       & 3.53 \\
          & 50   & 7     & 814.75 &       & 10    & 187.30 &       & 2.47 \\
          & 60   & 7     & 735.94 &       & 10    & 172.99 &       & 2.19 \\
    500   & 30   & 0     & 1200.00 &       & 10    & 174.39 &       & 4.38 \\
          & 40   & 0     & 1200.00 &       & 10    & 341.53 &       & 3.49 \\
          & 50   & 0     & 1200.00 &       & 10    & 457.51 &       & 3.94 \\
          & 60   & 7     & 866.69 &       & 10    & 396.45 &       & 3.71 \\
    600   & 30   & 0     & 1200.00 &       & 10    & 176.77 &       & 4.36 \\
          & 40   & 0     & 1200.00 &       & 10    & 333.33 &       & 3.46 \\
          & 50   & 0     & 1200.00 &       & 10    & 467.29 &       & 3.34 \\
          & 60   & 0     & 1200.00 &       & 10    & 537.73 &       & 4.10 \\
    \midrule
    \multicolumn{2}{c}{Overall} & 21    & 1101.45 &       & 120   & 304.55 &       & 3.62 \\
    \bottomrule
    \end{tabular}%
  \label{tab:TableAgile250}%
\end{table}%

% Table generated by Excel2LaTeX from sheet 'Sheet2'
\begin{table}[p]
  \centering
    \footnotesize
  \caption{Computational results of CG-LAB and CG-CPL on instances with 300 targets.}
    \begin{tabular}{ccrrrrrrr}
    \toprule
          &       & \multicolumn{2}{c}{CG-CPL} &       & \multicolumn{2}{c}{CG-LAB} &       &  \\
\cmidrule{3-4}\cmidrule{6-7}    $M^{C}$ \blue{(MB)} & $E^{C}$ \blue{(kJ)}& \multicolumn{1}{c}{$ub$ $( /10)$} & \multicolumn{1}{c}{$time$}\blue{(s)} &       & \multicolumn{1}{c}{$ub$ $( /10)$} & \multicolumn{1}{c}{$time$}\blue{(s)} &       & \multicolumn{1}{c}{$gap$ $(\%)$} \\
 \midrule
    400   & 30   & 0     & 1200.00 &       & 10    & 224.15 &       & 4.08 \\
          & 40   & 0     & 1200.00 &       & 10    & 313.74 &       & 3.47 \\
          & 50   & 3     & 999.46 &       & 10    & 255.57 &       & 2.96 \\
          & 60   & 3     & 980.87 &       & 10    & 249.29 &       & 2.53 \\
    500   & 30   & 0     & 1200.00 &       & 10    & 219.12 &       & 4.52 \\
          & 40   & 0     & 1200.00 &       & 10    & 417.71 &       & 3.76 \\
          & 50   & 0     & 1200.00 &       & 10    & 601.15 &       & 4.57 \\
          & 60   & 2     & 1070.50 &       & 10    & 494.25 &       & 4.29 \\
    600   & 30   & 0     & 1200.00 &       & 10    & 199.31 &       & 4.35 \\
          & 40   & 0     & 1200.00 &       & 10    & 351.59 &       & 3.61 \\
          & 50   & 0     & 1200.00 &       & 9     & 700.49 &       & 3.90 \\
          & 60   & 0     & 1200.00 &       & 9     & 862.59 &       & 4.78 \\
    \midrule
    \multicolumn{2}{c}{Overall} & 8     & 1154.24 &       & 118   & 407.41 &       & 3.90 \\
    \bottomrule
    \end{tabular}%
  \label{tab:TableAgile300}%
\end{table}%

Overall, CG-LAB consistently outperforms CG-CPL, obtaining the upper bound for all instances from 150 to 250 targets and only failing for two instances with 300 targets.  CG-CPL, on the other hand, already starts to experience difficulties on the smallest instances with 150 targets. As the number of targets increases, the performance of CG-CPL becomes worse, in that the LP is usually not solved within the time limit. With 300 targets, CG-CPL can only solve eight out of 120 instances. The cause of CG-CPL's failure is the time required for the pricing problems. %: . It costs longer time for the solver to trace a shortest path compared with the label-setting algorithm.
Although the label-setting pricing solver typically finds a shortest path very efficiently, CG-LAB fails to provide the upper bound on two instances with 300 targets; in these instances the number of mission paths for the pricing problem is close to one million, and the label-setting procedure also runs into difficulties. After the CG-procedure, the integer master problem with all generated columns is easily solved in less than one second in all cases.

As for the quality of the scheduling solutions, the average relative gap of CG-LAB is less than~$3\%$. As the number of targets increases, the gap from the upper bound rises slightly but never exceeds $5\%$ for any instance. It is worth pointing out that the actual optimality gap for the solution produced by CG-LAB is typically smaller than the $gap$ value, while a guaranteed optimal solution cannot be obtained in limited time. %This verifies that the bound obtained by the CG based algorithm is indeed quite tight.

The memory capacity $M^{C}$ significantly influences algorithmic performance. Larger $M^{C}$ provides more possibilities to execute more observation missions, while it also requires more  time to solve the pricing problem for both of the algorithms. For CG-LAB, longer running times  are required as $M^{C}$ increases, but the upper bound is still obtained for most instances, while CG-CPL gets stuck in the pricing problem due to the runtime limitations.

Different energy capacity values $E^{C}$ have different impact on the computational performance. For smaller $E^{C}$, the number of selected observation missions is low and fewer paths are generated, so a shorter running time is therefore needed for solving the pricing problem. When $E^{C}$ increases, on the other hand, a decrease in running time is sometimes also observed when larger $E^{C}$ no longer restrains the number of observation missions, so that the relaxed energy constraint is never binding anymore and thus does not have much impact anymore  on the runtime.

\subsubsection{Results for CEOS instances}
The CG-based framework is developed for scheduling agile satellites, so it can also be applied for planning conventional satellites. By fixing a satellite along the pitch axis, we transform an AEOS into a CEOS\@. The same target distribution and orbital parameters of the satellite constellation are considered.  We compare CG-LAB, CG-CPL and the compact flow formulation FF; the results are represented in Tables~\ref{tab:TableNonAgile150} to~\ref{tab:TableNonAgile300}.  The columns \textit{opt-gap} report the gap between the solution of CG-LAB %/CPL 
and the optimal integer solution obtained by FF, which is always readily available. On the instances with 150, 200, and 250 targets, FF needs less than one second on average for producing an optimal solution. CG-LAB always outperforms CG-CPL because of the pricing solver, and the size of the instance does not have a strong impact on its running time. Even for the largest instance with 300 targets, CG-LAB only requires 3.35 seconds on average, while over 60 seconds are needed for FF.

The solutions produced by the CG-heuristic are near-optimal. Compared with the LP-bound, the gap is less than $1\%$ on average and never exceeds $3\%$ for any instance.  For these conventional instances, we can also compute the actual optimality gap \textit{opt-gap}, which is even more favorable. For the instances with 150 and 200 targets, CG-LAB always finds an optimal solution. For the larger instances, the optimality gap is less than $0.1\%$. Overall, FF can efficiently find an optimal solution for small CEOS instances, while CG-LAB produces near-optimal solutions for larger instances with less runtime than FF.
% Our proposed CG based framework has superior performance for the CEOS scheduling.
% Table generated by Excel2LaTeX from sheet 'TABLE'

\begin{table}[p]
  \centering
    \footnotesize
  \caption{Computational results for CEOS instances with 150 targets.}
    \begin{tabular}{ccccrrr}
    \toprule
          &       & CG-LAB & CG-CPL & \multicolumn{1}{c}{FF} &       &  \\
\cmidrule{3-5}    $M^{C}$ \blue{(MB)} & $E^{C}$ \blue{(kJ)} & $time$ \blue{(s)} & $time$ \blue{(s)} & \multicolumn{1}{c}{$time$}\blue{(s)} & \multicolumn{1}{c}{$gap$ $(\%)$} & \multicolumn{1}{c}{\textit{opt-gap} $(\%)$} \\
    \midrule
    400   & 30   & 2.25  & 9.50  & 0.10  & 0.62  & 0.00 \\
          & 40   & 2.30  & 8.69  & 0.09  & 0.39  & 0.00 \\
          & 50   & 2.34  & 8.31  & 0.06  & 0.25  & 0.00 \\
          & 60   & 2.39  & 7.65  & 0.05  & 0.16  & 0.00 \\
    500   & 30   & 2.31  & 9.49  & 0.09  & 0.62  & 0.00 \\
          & 40   & 2.36  & 9.05  & 0.10  & 0.41  & 0.00 \\
          & 50   & 2.39  & 8.70  & 0.07  & 0.28  & 0.00 \\
          & 60   & 2.37  & 7.65  & 0.05  & 0.16  & 0.00 \\
    600   & 30   & 2.39  & 9.18  & 0.09  & 0.62  & 0.00 \\
          & 40   & 2.35  & 8.52  & 0.10  & 0.41  & 0.00 \\
          & 50   & 2.47  & 9.02  & 0.07  & 0.27  & 0.00 \\
          & 60   & 2.45  & 8.05  & 0.05  & 0.16  & 0.00 \\
    \midrule
    \multicolumn{2}{c}{Overall} & 2.36  & 8.65  & 0.08  & 0.36  & 0.00 \\
    \bottomrule
    \end{tabular}%
  \label{tab:TableNonAgile150}%
\end{table}%

% Table generated by Excel2LaTeX from sheet 'TABLE'
\begin{table}[p]
  \centering
    \footnotesize
  \caption{Computational results for CEOS instances with 200 targets.}
    \begin{tabular}{ccccrrr}
    \toprule
          &       & \multicolumn{1}{c}{CG-LAB} & \multicolumn{1}{c}{CG-CPL} & \multicolumn{1}{c}{FF} &       &  \\
\cmidrule{3-5}    $M^{C}$ \blue{(MB)} & $E^{C}$ \blue{(kJ)} & $time$ \blue{(s)} & $time$ \blue{(s)} & \multicolumn{1}{c}{$time$}\blue{(s)} & \multicolumn{1}{c}{$gap$ $(\%)$} & \multicolumn{1}{c}{\textit{opt-gap} $(\%)$} \\
    \midrule
    400   & 30   & 2.79  & 12.78 & 0.60  & 1.21  & 0.00 \\
          & 40   & 2.75  & 11.78 & 0.46  & 0.69  & 0.00 \\
          & 50   & 2.80  & 10.60 & 0.21  & 0.44  & 0.00 \\
          & 60   & 2.76  & 10.23 & 0.11  & 0.37  & 0.00 \\
    500   & 30   & 2.72  & 12.71 & 0.59  & 1.21  & 0.00 \\
          & 40   & 2.87  & 12.31 & 0.48  & 0.69  & 0.00 \\
          & 50   & 2.96  & 11.14 & 0.29  & 0.44  & 0.00 \\
          & 60   & 2.90  & 10.78 & 0.20  & 0.33  & 0.00 \\
    600   & 30   & 2.71  & 11.61 & 0.59  & 1.21  & 0.00 \\
          & 40   & 0.69  & 11.45 & 0.53  & 2.80  & 0.00 \\
          & 50   & 0.46  & 11.54 & 0.38  & 2.97  & 0.00 \\
          & 60   & 0.34  & 11.15 & 0.26  & 2.98  & 0.00 \\
    \midrule
    \multicolumn{2}{c}{Overall} & 2.23  & 11.51 & 0.39  & 1.28  & 0.00 \\
    \bottomrule
    \end{tabular}%
  \label{tab:TableNonAgile200}%
\end{table}%

% Table generated by Excel2LaTeX from sheet 'TABLE'
\begin{table}[p]
  \centering
    \footnotesize
  \caption{Computational results for CEOS instances with 250 targets.}
    \begin{tabular}{ccccrrr}
    \toprule
          &       & \multicolumn{1}{c}{CG-LAB} & \multicolumn{1}{c}{CG-CPL} & \multicolumn{1}{c}{FF} &       &  \\
\cmidrule{3-5}    $M^{C}$ \blue{(MB)} & $E^{C}$ \blue{(kJ)} & $time$ \blue{(s)} & $time$ \blue{(s)} & \multicolumn{1}{c}{$time$}\blue{(s)} & \multicolumn{1}{c}{$gap$ $(\%)$} & \multicolumn{1}{c}{\textit{opt-gap} $(\%)$} \\
    \midrule
    400   & 30   & 2.53  & 11.92 & 0.75  & 1.25  & 0.00 \\
          & 40   & 2.60  & 12.00 & 0.38  & 0.80  & 0.00 \\
          & 50   & 2.57  & 10.97 & 0.17  & 0.47  & 0.01 \\
          & 60   & 2.65  & 10.49 & 0.11  & 0.44  & 0.01 \\
    500   & 30   & 2.60  & 12.94 & 0.55  & 1.27  & 0.00 \\
          & 40   & 2.78  & 12.97 & 0.45  & 0.83  & 0.00 \\
          & 50   & 2.72  & 12.51 & 0.19  & 0.54  & 0.00 \\
          & 60   & 2.80  & 10.97 & 0.14  & 0.36  & 0.01 \\
    600   & 30   & 2.62  & 12.72 & 0.69  & 1.27  & 0.01 \\
          & 40   & 2.74  & 15.62 & 0.49  & 0.84  & 0.00 \\
          & 50   & 2.75  & 12.02 & 0.23  & 0.57  & 0.00 \\
          & 60   & 2.84  & 11.60 & 0.15  & 0.39  & 0.00 \\
    \midrule
    \multicolumn{2}{c}{Overall} & 2.68  & 12.23 & 0.36  & 0.72  & 0.00 \\
    \bottomrule
    \end{tabular}%
  \label{tab:TableNonAgile250}%
\end{table}%

% Table generated by Excel2LaTeX from sheet 'TABLE'
\begin{table}[p]
  \centering
    \footnotesize
  \caption{Computational results for CEOS instances with 300 targets.}
    \begin{tabular}{ccccrrr}
    \toprule
          &       & \multicolumn{1}{c}{CG-LAB} & \multicolumn{1}{c}{CG-CPL} & \multicolumn{1}{c}{FF} &       &  \\
\cmidrule{3-5}    $M^{C}$ \blue{(MB)} & $E^{C}$ \blue{(kJ)} & $time$ \blue{(s)} & $time$ \blue{(s)} & \multicolumn{1}{c}{$time$} \blue{(s)} & \multicolumn{1}{c}{$gap$ $(\%)$} & \multicolumn{1}{c}{\textit{opt-gap} $(\%)$} \\
    \midrule
    400   & 30   & 3.16  & 15.03 & 106.88 & 1.52  & 0.05 \\
          & 40   & 3.31  & 13.94 & 8.99  & 0.93  & 0.11 \\
          & 50   & 3.15  & 12.50 & 1.67  & 0.54  & 0.08 \\
          & 60   & 3.21  & 12.37 & 0.61  & 0.45  & 0.08 \\
    500   & 30   & 3.20  & 14.80 & 123.11 & 1.57  & 0.03 \\
          & 40   & 3.35  & 14.93 & 122.09 & 1.14  & 0.05 \\
          & 50   & 3.35  & 13.87 & 9.94  & 0.63  & 0.11 \\
          & 60   & 3.44  & 13.85 & 1.73  & 0.43  & 0.10 \\
    600   & 30   & 3.16  & 13.88 & 123.48 & 1.64  & 0.11 \\
          & 40   & 3.44  & 16.39 & 122.16 & 1.18  & 0.07 \\
          & 50   & 3.66  & 16.08 & 100.54 & 0.77  & 0.08 \\
          & 60   & 3.73  & 14.89 & 36.04 & 0.52  & 0.07 \\
    \midrule
    \multicolumn{2}{c}{Overall} & 3.35  & 14.38 & 63.10 & 0.84  & 0.08 \\
    \bottomrule
    \end{tabular}%
  \label{tab:TableNonAgile300}%
\end{table}%

\subsubsection{Sensitivity analysis} \label{subsubsec:discretization}

The discretization unit (the time between the start of two successive $OTW$s) of the $VTW$s is a key parameter of the model. In this section, we examine its influence on the performance of the proposed CG-based framework. We generate ten AEOS instances of 200 targets, combining 150 globally distributed targets and 50 targets in one interest area. The memory and energy capacity are fixed as 500 MB and 50 kJ, respectively. The simulation results for different discretization units are reported in Table~\ref{tab:SensiAnaly}, where column $dt$ represents the discretization unit in seconds and $|M|$ stands for the number of generated candidate observation missions. The columns $gap$ and $time$ contain the average value per setting only for the instances for which we obtain an optimal LP-solution.

% Table generated by Excel2LaTeX from sheet 'Sheet3'
\begin{table}[htbp]
	\centering
	  \footnotesize
	\caption{Performance of CG-LAB with different discretization units.}
	\begin{tabular}{rrrrr}
		\toprule
		\multicolumn{1}{c}{$dt$} \blue{(s)} & \multicolumn{1}{c}{$|M|$} & \multicolumn{1}{c}{$ub$ $( /10)$} & \multicolumn{1}{c}{$gap$ $(\%)$} & \multicolumn{1}{c}{$time$}\blue{(s)} \\
		\midrule
		1     & 22999 & 7     & 2.09  & 314.89 \\
		1.5   & 16695 & 10    & 2.28  & 203.09 \\
		2     & 12980 & 10    & 2.51  & 117.27 \\
		5     & 5242  & 10    & 2.25  & 30.66 \\
		10    & 2682  & 10    & 2.31  & 13.96 \\
		15    & 1828  & 10    & 1.48  & 9.15 \\
		20    & 1396  & 10    & 1.96  & 7.94 \\
		\bottomrule
	\end{tabular}%
	\label{tab:SensiAnaly}%
\end{table}%

The number of generated missions clearly rises as $dt$ decreases. When $dt = 1$ second, we cannot obtain the LP-bound for three out of the 10 instances because of an excessive number of missions. The gap from the LP-bound, on the other hand, is not very sensitive to $dt$, and is always below~$3\%$. The value of $dt$ does have a very serious impact on the CPU time, with lower choices for $dt$ obviously leading to longer runtimes. The FF model cannot find an optimal solution within the time limit even when $dt = 20$ seconds. In conclusion, the user should carefully select a proper value of $dt$ for the proposed model, in order to strike a balance between fine discretization and running time.
 %Therefore $Dt$ is set as 2 seconds in the above section of experiment setup.

\section{Conclusions}
\label{sec:conclusion}
In this article, we have studied the scheduling of observations by multiple agile satellites and with the possibility of conducting multiple observations for the same target. We aim to maximize the entire observation profit, with a profit function per target that can be nonlinear in the number of scheduled missions. We describe a linear flow-based formulation, which can handle sequence-dependent constraints but whose computational performance turns out to be limited.

We also develop a CG-based framework to solve a decomposition of  the flow formulation. A label-setting algorithm is introduced for solving the pricing problems. We compare the flow model and two CG-heuristics, in which the pricing problem is solved by a label-setting algorithm and by an IP-solver, respectively. The computational results indicate that the flow formulation efficiently obtains optimal solutions on small instances for conventional satellite scheduling, while the CG-heuristic has superior performance for scheduling agile satellites, and the dedicated pricing solver is clearly better than the IP-solver for pricing. The proposed CG-framework also provides a tight upper bound that allows to effectively evaluate the quality of heuristic solutions.

A topic for future research can be to adjust the objective function to incorporate completion times for user-required missions, since rapid response plays a very important role in many observation missions, for instance in case of natural disasters. \blue{In addition, future work could focus on incorporating further operational factors such as cloud uncertainty and telemetry, tracking, and command scheduling, as well as modeling data transmission limitations more accurately. One possible approach is to represent downlink opportunities as dynamic time-dependent resource constraints~\citep{xiang2024hierarchical}. These extensions will further improve the realism and applicability of the proposed framework.}

% A topic for future research can be to adjust the objective function to incorporate completion times for user-required missions, since rapid response plays a very important role in many observation missions, for instance in case of natural disasters.  A more accurate modelling of data transmission constraints is also an avenue for further research. % Finally, another planning complexity that has recently received some research attention but which deserves to be further explored, is related to the anticipation of uncertainty, for instance regarding cloud coverage. % Other opportunities for adding more constraints into the model, such as data transmission, into the AEOS scheduling are legion.

% \section*{Acknowledgment}
% This research was supported by the China Scholarship Council and the Academic Excellence Foundation of BUAA for PhD-students.

\begin{appendices}
%\section*{Appendix}
%\setcounter{subsection}{0}

\section{The label-setting procedure}
\label{ApenB}
Following~\citet{GabrelVanderpooten-61}, we apply a label-setting method for the pricing problems. Since each edge in each graph $G_{jk}$ leads to a mission with a strictly larger starting time than its origin node, we order the missions $M_{jk}$ in ascending observation starting time so that all paths are explored by the loop between lines \ref{L31} and \ref{L32} in Algorithm~\ref{Algorithm3}.

For each mission $q$ in $G_{jk}$, we denote the set of possible predecessors as $M^{-q}_{jk}$ and employ function $GetCurPreMisIndex()$ to obtain the mission index of the current predecessor. Each path associated with mission $r$ in $\mathscr{P}_{jkr}$ is extended by adding the current $MisCost_{jkq}$, $MisM_{jkq}$ and $MisE_{jkq}$. Function $CheckConstraints()$ returns $true$ when $NewP$ satisfies all constraints and returns $false$ otherwise. Function $CheckPathDominance()$ returns $true$ if $NewP$ is not dominated by any paths in $\mathscr{P}_{jkq}$ and returns $false$ otherwise. Subsequently, $NewP$ can be added into $\mathscr{P}_{jkq}$ if $NewP$ is not dominated. We also check whether $NewP$ dominates any paths in $\mathscr{P}_{jkq}$, which can then be removed. Eventually, we obtain a path from source to sink with minimal cost.
%If the minimal cost is less than 0, a column is generated according to the missions along the shortest path.

{The worst-case complexity of Algorithm~\ref{Algorithm3} is exponential}%, where $n$ is the number of vertices in the graph and exponential runtimes are unavoidable
, knowing that the resource-constrained shortest path problem is NP-hard even with one resource \citep{MehlhornZiegelmann}. In practice, however, the number of dominant paths may be relatively low thanks to the resource constraints and dominance criteria \citep{GabrelVanderpooten-61}, and the foregoing procedure empirically turns out to be a computationally viable alternative. % $\mathscr{P}_{jke_{jk}}$ is in general tractable.

\begin{algorithm}[hbtp]
	\caption{The label-setting procedure}
	\label{Algorithm3}
	\begin{algorithmic}[1]
		\REQUIRE  %算法的输入参数：Input
		candidate observation missions set $M_{jk}$ ordered in ascending observation starting time, possible predecessor set $M^{-q}_{jk}$ for each mission and other associated parameters;\\
		\ENSURE  %算法的输出：Output
		%a minimal-cost path in $\mathscr{P}_{jke_{jk}}$; %
%        minimum of objective function \eqref{PriObjFunc};
   an optimal solution to model \eqref{PriObjFunc}--\eqref{NewEnergyConCG};
        %\STATE Order missions $M_{jk}$ with ascending observation starting time;\
		\FOR{each mission $q \in M_{jk} \cup
			\{e_{jk}\}$}
		\label{L31}
		\STATE $P_{jkq} = \emptyset;$ \
		\FOR{each possible predecessor mission in $M^{-q}_{jk}$}
		\STATE $r = GetCurPreMisIndex();$\
		\FOR{each $P_{jkrt^{\prime}} \in \mathscr{P}_{jkr}$ }
		\STATE$NewP = (Cost_{jkrt^{\prime}}+MisCost_{jkq},CurM_{jkrt^{\prime}} + MisM_{jkq},CurE_{jkrt^{\prime}}+MisE_{jkq})$;\
		\IF{ $CheckConstraints(NewP)$}
		\IF{ $CheckPathDominance(\mathscr{P}_{jkq},NewP)$}
		\STATE$ \mathscr{P}_{jkq} = \mathscr{P}_{jkq} \cup \{NewP\};$\
		\ENDIF
		\ENDIF
		\ENDFOR
		\ENDFOR
		\ENDFOR
		\label{L32}
		\RETURN  a minimal-cost path in $\mathscr{P}_{jke_{jk}}$
	\end{algorithmic}
\end{algorithm}

\section{Satellite parameters}
\label{ApenA}
Table~\ref{tab:SatParameters} contains the following details for the four satellites in \textit{Gaojing-1}.  The first column ID indicates the name of satellite, and the parameters in the remaining columns are the satellite's semi-major axis $a$, inclination $i$, right ascension of the ascending node $\Omega$, eccentricity $e$, argument of perigee $\omega$ and mean anomaly $M$, respectively. In addition, the agile platform of the constellation allows up to $30^{\circ}$ maneuvers along both the roll and the pitch axis.

\begin{table}[H]
	\caption{Orbital parameters of the satellite constellation \textit{Gaojing-1}.}
	\centering
	\label{tab:SatParameters}
\begin{tabular}{ccccccc}
  \hline
  % after \\: \hline or \cline{col1-col2} \cline{col3-col4} ...
  \toprule
    ID & $a$ (km) & $i$ ($^{\circ}$)  & $\Omega$ ($^{\circ}$) & $e$ & $\omega$ ($^{\circ}$) & $M$ ($^{\circ}$) \\
    \midrule
  $Sat1$ & 6903.673 & 97.5839 & 97.8446 & 0.0016546 & 50.5083 & 2.0288 \\
  $Sat2$ & 6903.730 & 97.5310 & 95.1761 & 0.0015583 & 52.2620 & 31.4501 \\
  $Sat3$ & 6909.065 & 97.5840 & 93.1999 & 0.0009966 & 254.4613 & 155.2256 \\
  $Sat4$ & 6898.602 & 97.5825 & 92.3563 & 0.0014595 & 276.7332 & 140.1878 \\
  \bottomrule
\end{tabular}
\end{table}

\end{appendices}

\bibliographystyle{apalike}
% \bibliography{Main}
\bibliography{SatScheduling-final}
%\bibliography{CG}
\end{document}